\documentclass[a4paper,UKenglish]{lipics-v2016}
\usepackage{cite} 

\newcommand{\OMIT}[1]{}

\newcommand{\D}{{\cal D}} 
\newcommand{\G}{{\cal G}}

\newcommand{\np}{{\rm NP}}
\newcommand{\conexp}{{\rm coNEXPTIME}} 
\newcommand{\nexp}{{\rm NEXPTIME}} 
\newcommand{\exptime}{{\rm EXPTIME}} 
\newcommand{\expspace}{{\rm EXPSPACE}} 
 
\newcommand{\coNP}{{\rm coNP}} 
 
\newcommand{\CQ}{{\sf CQ}} 
\newcommand{\UCQ}{{\sf UCQ}} 
\newcommand{\ptime}{{\rm P}} 
\newcommand{\CRPQ}{{\sf CRPQ}} 
\newcommand{\UCRPQ}{{\sf UCRPQ}} 
\newcommand{\CTW}{{\sf TW}_{\rm crpq}} 
\newcommand{\TW}{{\sf TW}}  
\newcommand{\UTW}{{\sf UTW}}  

\newcommand{\lab}{{\sf label}}

\newcommand{\C}{{\cal C}} 
\renewcommand{\L}{{\cal L}}

\theoremstyle{plain} 
\newtheorem{proposition}{Proposition}

\usepackage{microtype}%if unwanted, comment out or use option "draft"

\title{The Complexity of Reverse Engineering Problems for Conjunctive Queries}
\titlerunning{The Complexity of Reverse Engineering Problems for Conjunctive Queries} %optional, in case that the title is too long; the running title should fit into the top page column

\author{Pablo Barcel\'o}
\author{Miguel Romero}
\affil{Center for Semantic Web Research \& \\ 
Department of Computer Science, 
  University of Chile\\
  \texttt{pbarcelo@dcc.uchile.cl, mromero@dcc.uchile.cl}}
\authorrunning{P. Barcel\'o and M. Romero} 
\Copyright{Pablo Barcel\'o and Miguel Romero}%mandatory, please use full first names. LIPIcs license is "CC-BY";  http://creativecommons.org/licenses/by/3.0/

% mandatory: Please provide 1-5 keywords
\begin{document}

\maketitle

\begin{abstract}
Reverse engineering problems for conjunctive queries (CQs), such as query by example (QBE) or definability, take a set of user 
examples and convert them into an explanatory CQ. Despite their importance, the complexity
of these problems is prohibitively high (\conexp-complete). We isolate their two main sources of 
complexity and propose relaxations of them that reduce the complexity while having meaningful theoretical 
interpretations. The first relaxation 
is based on the idea of using existential pebble games for approximating homomorphism tests. We show that this 
characterizes QBE/definability for CQs up to treewidth $k$, while reducing the complexity to \exptime. 
As a side result, we obtain that 
the complexity of the QBE/definability problems for CQs of treewidth $k$ is \exptime-complete for each $k \geq 1$. 
The second relaxation is based on the idea of ``desynchronizing'' direct products, which 
characterizes QBE/definability for unions 
of CQs and reduces the complexity to \coNP. 
The combination of these two relaxations yields tractability for QBE and characterizes it 
in terms of unions of CQs of treewidth at most $k$. 
We also study the complexity of these problems for conjunctive regular path queries over 
graph databases, showing them to be no more difficult than for CQs.  
 \end{abstract}

\section{Introduction}

{\em Reverse engineering} is the general problem of abstracting user examples into an explanatory 
query. An important instance of this problem corresponds to {\em query-by-example} (QBE) for a 
query language $\L$. In QBE, the system is presented with a 
database $\D$ and $n$-ary relations $S^+$ and $S^-$  over $\D$ of {\em positive} and {\em negative} 
examples, respectively. The question is whether there exists  
a query $q$ in $\L$ such that its evaluation $q(\D)$ over $\D$ contains all the positive examples (i.e., $S^+ \subseteq q(\D)$)
but none of the negative ones (i.e., $q(\D) \cap S^- = \emptyset$). In case such $q$ exists, it is also 
desirable to return its result $q(\D)$. Another version of this problem assumes that the system is given the set 
$S^+$ of positive examples 
only, and the question is whether there is a query $q$ in $\L$ that precisely defines $S^+$, i.e., $q(\D) = S^+$.  
This is often known as the {\em definability problem} for $\L$. 
As of late, QBE and definability have received quite some attention in different contexts; e.g., 
for first-order logic and the class of conjunctive queries over relational databases 
\cite{ZEPS13,TCP14,LCM15,BCS16,AD16,Willard10,CD15}; for regular path queries over graph databases \cite{ANS13,BCL15};  
for SPARQL queries over RDF \cite{ADK16}; and for tree patterns over XML \cite{CW15,SW15}.  

In data management, a particularly important instance of QBE and definability  
corresponds to the case when 
$\L$ is the class of conjunctive queries (CQs). 
Nevertheless, the relevance of such instance is counterbalanced
 by its inherent complexity: Both QBE and 
definability for CQs are \conexp-complete \cite{Willard10,CD15}. Moreover, in case that a CQ-{\em explanation} 
$q$ for $S^+$ and $S^-$ over $\D$ exists (i.e., a CQ $q$ such that $S^+ \subseteq q(\D)$ and 
$q(\D) \cap S^- = \emptyset$ for QBE), 
it might take double exponential time to compute its result $q(\D)$. 
While several heuristics have been proposed that alleviate this complexity in practice 
\cite{ZEPS13,TCP14,LCM15,BCS16}, up to date there has been (essentially) no theoretical investigation  
identifying the sources of complexity of these problems and proposing principled solutions for them. The 
general objective of this article is to make a first step in such direction.   

A semantic characterization of QBE for CQs has been known for a long time in the community. 
Formally, there exists a CQ $q$ such that 
$S^+ \subseteq q(\D)$ and 
$q(\D) \cap S^- = \emptyset$ (i.e., a CQ-explanation) 
if and only if (essentially) the following {\em QBE test for CQs} 
succeds:

\begin{itemize} 
\item \underline{QBE test for CQs:}
For each 
tuple $\bar b$ in $S^-$ it is the case that 
$\prod_{\bar a \in S^+} (\D,\bar a) \not\to (\D,\bar b)$, i.e.,  $\prod_{\bar a \in S^+} (\D,\bar a)$ does not 
homomorphically map 
to 
$(\D,\bar b)$. 
(Here, $\prod$ denotes the usual direct 
product of databases with distinguished tuples of constants).
\end{itemize}  

 \noindent 
(A similar test characterizes CQ-definability, save that now $\bar b$ is an arbitrary tuple over $\D$ outside $S^+$).  
Moreover, in case there is a CQ-explanation $q$ for $S^+$ and $S^-$ over $\D$, 
then there is a {\em canonical} such explanation
given by the CQ whose body corresponds to $\prod_{\bar a \in S^+} (\D,\bar a)$. As shown by Willard \cite{Willard10}, 
the QBE test for CQs yields optimal bounds for 
determining (a) the existence of a CQ-explanation $q$ for $S^+$ and $S^-$ over $\D$ 
(namely, \conexp), and (b) the size of such $q$ (i.e., exponential). More important,  
it allows to identify the two main sources of complexity of the problem, each one of which increases its complexity by one exponential: 
\begin{enumerate}
\item The construction of the canonical explanation $\prod_{\bar a \in S^+} (\D,\bar a)$, which takes exponential time  
in the combined size of $\D$ and $S^+$. 
\item The homomorphism test 
$\prod_{\bar a \in S^+} (\D,\bar a) \to (\D,\bar b)$ for each tuple $\bar b \in S^-$. 
Since, in general, checking for the existence of a homomorphism is an \np-complete problem, this 
step involves an extra exponential blow up.  
\end{enumerate}  

\smallskip 
\noindent
{\bf Our contributions:} We propose relaxations of the QBE test for CQs that alleviate one or 
both sources of complexity and have  
meaningful theoretical interpretations in terms of the QBE problem (our results also apply to definability). 
They are based on standard approximation notions for the homomorphism test and 
the construction of the direct product $\prod_{\bar a \in S^+} (\D,\bar a)$, as found in the context of constraint 
satisfaction and definability, respectively. 

\begin{enumerate}

\item We start by relaxing the second source of complexity, i.e., 
the one given by the homomorphism tests of the form $\prod_{\bar a \in S^+} (\D,\bar a) \to (\D,\bar b)$, 
for $\bar b \in S^-$. In order to approximate the notion of homomorphism, we use the 
{\em strong consistency tests} often applied in the area of constraint satisfaction \cite{Dechter92}. 
As observed by Kolaitis and Vardi \cite{KV-aaai}, such consistency tests can be recast in terms of 
the {\em existential pebble game} \cite{KV95}, first defined in the context of database theory as 
a tool for studying the expressive power of Datalog, and also used 
to show that CQs of bounded treewidth can be evaluated efficiently \cite{DKV02}. 

As opposed to the homomorphism test, checking for the existence of a 
{\em winning duplicator strategy} 
in the existential $k$-pebble game on $(\D,\bar a)$ and $(\D',\bar b)$, denoted $(\D,\bar a) \to_k (\D',\bar b)$, 
can be solved in polynomial time for each fixed $k > 1$ \cite{KV95}. 
Therefore, replacing 
the homomorphism test 
$\prod_{\bar a \in S^+} (\D,\bar a) \to (\D,\bar b)$ with its ``approximation''   
$\prod_{\bar a \in S^+} (\D,\bar a) \to_{k} (\D,\bar b)$ 
reduces the complexity of the QBE test for CQs to \exptime. Furthermore, this approximation 
has a neat theoretical interpretation: The relaxed version of the 
QBE test accepts the input given by $(\D,S^+,S^-)$ if and only if there is a CQ-explanation
$q$ for $S^+$ and $S^-$ over $\D$ such that $q$ is of {\em treewidth at most $(k-1)$}. 
While the latter is not particularly surprising in light of the strong existing connections between the existential $k$-pebble game and 
the evaluation of CQs of treewidth at most $(k-1)$ \cite{DKV02}, we believe our characterization 
to be of conceptual importance.  

Interestingly, when this relaxed version of the QBE test yields a CQ-explanation $q$ of treewidth at most 
$(k-1)$, its 
result $q(\D)$ can be evaluated in exponential time (recall that for general CQs this might require double exponential time).  

\item We then prove that the previous bound is optimal, i.e., checking whether the relaxed version of the 
QBE test accepts the input given by $(\D,S^+,S^-)$, or, equivalently, if there is a CQ-explanation  
$q$ for $S^+$ and $S^-$ over $\D$ of treewidth at most $k$, for each $k \geq 1$, is 
\exptime-complete. (This also holds for the definability problem for CQs of treewidth at most $k$). 
Intuitively, this states that 
relaxing the second source of complexity of the test by using existential pebble games does not 
eliminate the first one. 

\item Finally, we look at the second source of complexity, i.e., the construction of the exponential size 
canonical explanation $\prod_{\bar a \in S^+} (\D,\bar a)$. While it is not clear which techniques are better suited 
for approximating this construction, we look at a particular one that appears in the context of 
definability: Instead of constructing the synchronized product $\prod_{\bar a \in S^+} (\D,\bar a)$ with respect to all tuples in $S$, 
we look at them one by one. That is, we check whether for each 
tuples $\bar a \in S^+$ and $\bar b \in S^-$ it is the case that $(\D,\bar a) \not\to (\D,\bar b)$. By using a 
characterization developed in the context of definability \cite{ANS13}, we observe that this relaxed version of the QBE test is 
\coNP-complete and has a meaningful interpretation: It corresponds to finding explanations 
based on {\em unions} of CQs. Moreover, when combined with the previous relaxation (i.e., replacing 
the homomorphism test $(\D,\bar a) \to (\D,\bar b)$ with $(\D,\bar a) \to_k (\D,\bar b)$) 
we obtain tractability. This further relaxed test corresponds to  finding explanations over the set of unions of CQs
of treewidth at most $(k-1)$. 

\end{enumerate} 

We then switch to study QBE in the context of graph databases, where 
CQs are often extended with the ability to check whether 
two nodes are linked by a path whose label satisfies a given regular expression. 
This gives rise to the class of {\em conjunctive regular path queries}, or 
CRPQs (see, e.g., \cite{CM90,CGLV00,Woo,Bar13}). 
CRPQ-definability was first studied
by Antonopulos et al. \cite{ANS13}. In particular, it is shown that CRPQ-definability is in \expspace\ 
by exploting automata-based 
techniques, in special, pumping arguments. Our contributions in this context 
are the following: 

\begin{enumerate} 
\item 
We first provide a QBE test for CRPQs
in the spirit of the one
for CQs given above. With such characterization we prove that 
QBE and definability for CRPQs are in \conexp, improving the \expspace\ upper bound of Antonopoulos et al. 
This tells us that these problems are at least not more difficult than for CQs. 
\item 
We also develop relaxations of the QBE test for CRPQs based on the existential pebble game and 
the ``desynchronization'' of the direct product $\prod_{\bar a \in S^+} (\D,\bar a)$. 
As before, we show that they reduce the complexity of the test and have meaningful 
interpretations in terms of the class of queries we use to construct explanations. 
\end{enumerate} 

\medskip
\noindent
{\bf Organization:} Preliminaries are  in Section \ref{sec:prelim}. A review of QBE/definability for CQs is provided in Section \ref{sec:cq}. Relaxations of the homomorphism tests 
are studied in Section \ref{sec:tw} and the desynchronization of the direct product in Section \ref{sec:des}. 
In Section \ref{sec:ext} we consider QBE/definability for CRPQs. Future work is presented in 
Section \ref{sec:conc}. 
%Due to space limitations, we relegate some proofs to the appendix.

\section{Preliminaries} \label{sec:prelim} 

\noindent
{\bf Databases, homomorphisms, and direct products.} 
A {\em schema} is a finite set of relation symbols, each one of which has an associated arity $n > 0$. 
A {\em database} over schema $\sigma$ is a finite 
set of atoms of the form $R(\bar a)$, where $R$ is a relation symbol in $\sigma$ of arity $n > 0$ and 
$\bar a$ is an $n$-ary tuple of constants. We slightly abuse notation, and sometimes write $\D$ 
also for the set of elements mentioned in $\D$. 

Let $\D$ and $\D'$ be databases 
over the same schema $\sigma$. A {\em homomorphism} from $\D$ to $\D'$ 
is a mapping $h$ from the elements of $\D$ to the elements of $\D'$ such that for every atom  
$R(\bar a)$ in $\D$ it is the case that $R(h(\bar a)) \in \D'$. 
We often need to talk about distinguished tuples 
of elements in databases. We then write $(\D,\bar a)$ to define the pair that corresponds to the database $\D$ and 
the tuple $\bar a$ of elements in $\D$. Let $\bar a$ and $\bar b$ be $n$-ary 
($n \geq 0$)
tuples of elements  in $\D$ and $\D'$, respectively. A homomorphism from $(\D,\bar a)$ to $(\D',\bar b)$ 
is a homomorphism from $\D$ to $\D'$ such that $h(\bar a) = \bar b$. We write $(\D,\bar a) \to (\D',\bar b)$ if there is a homomorphism from
$(\D,\bar a)$ to $(\D',\bar b)$. Checking if $(\D,\bar a) \to (\D',\bar b)$ is a well-known \np-complete problem. 

In this work, the notion of {\em direct product} of databases is particularly important. 
Let $\bar a = (a_1,\dots,a_n)$ and $\bar b = (b_1,\dots,b_n)$ be $n$-ary tuples of elements over $A$ and $B$, respectively. 
Their direct product $\bar a \otimes \bar b$ is the $n$-ary tuple 
$((a_1,b_1),\dots,(a_n,b_n))$ over $A \times B$. If $\D$ and $\D'$ are databases over the same schema $\sigma$, 
we define $\D \otimes \D'$ to be the following database over $\sigma$:
$$\{R(\bar a \otimes \bar b) \, \mid \, R \in \sigma, \, R(\bar a) \in \D, \text{ and } R(\bar b) \in \D'\}.$$
Further, we use $(\D,\bar a) \otimes (\D',\bar b)$ to denote the pair $(\D \otimes \D',\bar a \otimes \bar b)$, and 
write $\prod_{1 \leq i \leq m} (\D_i,\bar a_i)$ as a shorthand 
for $(\D_1, \bar  a_1) \otimes \dots \otimes (\D_m,\bar a_m)$. This is allowed since 
$\otimes$ is an associative operation. 

The elements in the tuple $\prod_{1 \leq i \leq m} \bar a_i$ 
may or may not appear in $\prod_{1 \leq i \leq m} \D_i$. If they do appear, we  
call $\prod_{1 \leq i \leq m} (\D_i,\bar a_i)$ {\em safe}. 
The notion of safeness is important in our work for reasons that will become apparent later. 
The next example better explains this notion:

\begin{example} 
If $\D = \{R(a,b),S(c,d)\}$, 
$\bar a_1 = (a,b)$, and $\bar a_2 = (c,d)$, then $(\D,\bar a_1) \otimes (\D,\bar a_2)$ 
is unsafe. In fact, 
$\bar a_1 \otimes \bar a_2 = \big((a,c),(b,d)\big)$ and $\D \otimes \D = \{R((a,a),(b,b)),S((c,c),(d,d))\}$. 
That is, 
none of the elements in $\bar a_1 \otimes \bar a_2$ belongs to $\D \otimes \D$. \qed
\end{example} 

It is worth remarking that the direct product $\otimes$ defines the least upper bound in the lattice of databases defined by the notion of homomorphism. 
In particular: 
\begin{enumerate} 
\item $\prod_{1 \leq i \leq m} (\D_i,\bar a_i) \to (\D_i,\bar a_i)$ for each $1 \leq i \leq m$, and
 \item if $(\D,\bar a) \to (\D_i,\bar a_i)$ for each $1 \leq i \leq m$, then $(\D,\bar a) \to \prod_{1 \leq i \leq m} (\D_i,\bar a_i)$. 
\end{enumerate} 
   
\medskip 
\noindent
{\bf Conjunctive queries.} 
 A {\em conjunctive query} (CQ) $q$  
over relational schema $\sigma$ is an FO formula of the form: 
\begin{equation} \label{eq:cq} 
\exists \bar y \big(R_1(\bar x_1) \wedge \dots \wedge R_m(\bar x_m)\big),
\end{equation} 
such that (a) each $R_i(\bar x_i)$ is an atom over $\sigma$, for $1 \leq i \leq m$, and (b)
$\bar y$ is a sequence of variables taken from the $\bar x_i$'s. 
In order to ensure domain-independence for queries, we only consider CQs without constants. 
We often write $q(\bar x)$ to denote that $\bar x$
is the sequence of {\em free} variables of $q$, i.e., the ones that do not appear existentially quantified in $\bar y$.   

Let $\D$ be a database over $\sigma$. We define the evaluation of a CQ $q(\bar x)$ of the form 
\eqref{eq:cq} over $\D$ in terms of the homomorphisms from $\D_q$ to $\D$, where $\D_q$ 
is the {\em canonical database} of $q$, that is, $\D_q$ is the database 
$\{R_1(\bar x_1),\dots,R_m(\bar x_m)\}$ that contains all atoms in $q$. 
The {\em evaluation of $q(\bar x)$ 
over $\D$}, denoted $q(\D)$, contains exactly those tuples $h(\bar x)$ such that $h$ is a homomorphism 
from $\D_q$ to $\D$. 

\medskip
\noindent
{\bf CQs of bounded treewidth.} The evaluation problem for CQs (i.e., determining whether $q(\D) \neq \emptyset$, given a database $\D$ and a CQ $q$) is \np-complete, but becomes tractable for several syntactically defined classes. One of the most prominent such classes
corresponds to the CQs of {\em bounded treewidth} \cite{CR00}.   
Recall that treewidth is a graph-theoretical concept that measures how much a 
graph resembles a tree (see, e.g., \cite{diestel}). For instance, trees have treewidth one, cycles treewidth two, and 
$K_k$, the clique on $k$ elements, treewidth $k-1$.  

Formally, let $G = (V,E)$ be an undirected graph. 
A {\em tree decomposition} of $G$ is a pair $(T,\lambda)$, where $T$ is a tree and $\lambda$ is a mapping that assigns 
a nonempty set of nodes in $V$ to each node $t$ in $T$, for which the following holds: 
\begin{enumerate}
\item For each $v \in V$ it is the case that the set of nodes $t \in T$ such that $v \in \lambda(t)$ is connected. 
\item For each edge $\{u,v\} \in E$ there exists a node $t \in T$ such that $\{u,v\} \subseteq \lambda(t)$. 
\end{enumerate} 
The {\em width} of $(T,\lambda)$ corresponds to $(\max{\{|\lambda(t)| \mid t \in T\}}) - 1$. The treewidth of $G$ 
is then defined as the minimum width of its tree decompositions. 

We define the treewidth of a CQ $q = \exists \bar y \bigwedge_{1 \leq i \leq m} R_i(\bar x_i)$ 
as the treewidth of the {\em Gaifman graph} of its existentially quantified variables. 
Recall that this is the undirected graph 
whose vertices are the existentially quantified variables of $q$ (i.e., those in $\bar y$) and 
there is an edge between 
distinct existentially quantified variables $y$ and $y'$ if and only they appear together in some atom of $q$, that is, they both appear 
in a tuple $\bar x_i$ for $1 \leq i \leq m$. 
For $k \geq 1$, we denote by $\TW(k)$ the class of CQs of treewidth 
at most $k$. It is known that the evaluation problem for the class $\TW(k)$ (for each fixed $k \geq 1$) 
can be solved in polynomial time \cite{CR00,DKV02}. 

\medskip
\noindent 
{\bf The QBE and definability problems.} Let $\C$ be a class of queries 
(e.g., the class $\CQ$ of all conjunctive queries, or $\TW(k)$ of CQs of treewidth at most $k$). 
Suppose that $\D$ is a database and $S^+$ and $S^-$ are $n$-ary relations over $\D$
of positive and negative examples, respectively. A  {\em $\C$-explanation} for $S^+$ and $S^-$
over $\D$ is a query $q$ in $\C$ such that $S^+ \subseteq q(\D)$ and $q(\D) \cap S^- = \emptyset$. 
Analogously,  a {\em $\C$-definition} of $S^+$ over $\D$ is a query $q$ in $\C$ such that $q(\D) = S^+$. 
The {\em query by example} and {\em definability} problems for $\C$ 
are as follows: 

\begin{center}
\fbox{\begin{tabular}{ll}
\small{PROBLEM} : & {\sc $\C$-query-by-example} (resp., {\sc $\C$-definability})
\\{\small INPUT} : & A database $\D$ and $n$-ary relations $S^+$ and $S^-$ over $\D$ 
\\
& (resp., a database $\D$ and an $n$-ary relation $S^+$ over $\D$) \\ 
{\small QUESTION} : &  Is there a $\C$-explanation for $S^+$ and $S^-$
over $\D$?
\\
& (resp., is there a $\C$-definition of $S^+$ over $\D$?) 
\end{tabular}}
\end{center}

\section{Query by example and definability for CQs} 
\label{sec:cq} 

Let us start by recalling what is known about these problems for CQs. 
We first establish characterizations of the notions of $\CQ$-explanations/definitions 
based on the following tests: 

\begin{itemize} 
\item \underline{QBE test for CQs:} Takes as input a 
database $\D$ and $n$-ary relations $S^+$ and $S^-$ over $\D$. It 
accepts if and only if: 
\begin{enumerate} 
\item 
$\prod_{\bar a \in S^+} (\D,\bar a)$ is safe, and 
\item  
$\prod_{\bar a \in S^+} (\D,\bar a) \not\to (\D,\bar b)$ for each tuple $\bar b \in S^-$. 
\end{enumerate} 
\item \underline{Definability test for CQs:} Takes as input a 
database $\D$ and an $n$-ary relation $S^+$ over $\D$. It 
accepts if and only if:
\begin{enumerate} 
\item 
$\prod_{\bar a \in S^+} (\D,\bar a)$ is safe, and 
\item  
 $\prod_{\bar a \in S^+} (\D,\bar a) \not\to (\D,\bar b)$ for each $n$-ary tuple $\bar b$ over 
$\D$ that is not in $S^+$.
\end{enumerate} 
\end{itemize} 

The following characterizations are considered to be folklore in the literature:  

\begin{proposition} {\em \label{prop:folklore}} The following statements hold:
\begin{enumerate}
\item 
Let $\D$ be a database and $S^+,S^-$ relations over $\D$. There is a $\CQ$-explanation for 
$S^+$ and $S^-$ 
over $\D$ if and only if the QBE test for CQs accepts $\D$, $S^+$, and $S^-$.
\item 
Let $\D$ be a database and $S^+$ a relation over $\D$. There is a \CQ-definition for $S^+$
over $\D$ if and only if the definability test for CQs accepts $\D$ and $S^+$.
\end{enumerate} 
\end{proposition}  

This provides us with a simple method for obtaining  a \conexp\ upper bound for {\sc \CQ-query-by-example} and 
{\sc \CQ-definability}. Let us concentrate on the first problem (a similar argument works for the second one).  
Assume that $S^+$ and $S^-$ are relations of positive and negative examples over a database $\D$.  
It follows from Proposition \ref{prop:folklore} that to check that 
there is {\em not} \CQ-explanation for $S^+$ and $S^-$ over $\D$, we need to 
either show that $\prod_{\bar a \in S^+} (\D,\bar a)$ is unsafe or guess a tuple $\bar b \in S^-$ and a homomorphism 
$h$ from $\prod_{\bar a \in S^+} (\D,\bar a)$ to $(\D,\bar b)$. 
Since $\prod_{\bar a \in S^+} (\D,\bar a)$ is of exponential size, checking its safety can be carried out in exponential time. On the 
other hand, 
the guess of $h$ is also of exponential size, and therefore 
checking that $h$ is indeed a homomorphism from $\prod_{\bar a \in S^+} (\D,\bar a)$ to $(\D,\bar b)$ can be performed in 
exponential time. The whole procedure can then be carried out in \nexp. As it turns out, this bound is also optimal: 

\begin{theorem} {\em \cite{Willard10,CD15}} 
The problems {\sc \CQ-query-by-example} and {\sc \CQ-definability} are \conexp-complete. 
\end{theorem} 

The lower bound for {\sc CQ-definability} was established by 
Willard using a complicated reduction from the complement of a tiling problem. 
A simpler proof was then obtained by ten Cate and Dalmau \cite{CD15}.
Their techniques also establish a lower bound for {\sc \CQ-query-by-example}. 
Notably, these lower bounds hold even when $S^+$ and $S^-$ are unary relations. 

\medskip
\noindent
{\bf The cost of evaluating CQ-explanations.} Recall that in query by example not only we want to find a 
$\CQ$-explanation $q$ for $S^+$ and $S^-$ over $\D$, but also compute its result $q(\D)$ if possible. 
It follows from the proof of Proposition \ref{prop:folklore} that in case there is a $\CQ$-explanation for 
$S^+$ and $S^-$ over $\D$, then we can assume such CQ to be $\prod_{\bar a \in S^+} (\D,\bar a)$, i.e., 
the CQ whose set of atoms is 
$\D^{|S^+|}$ and whose tuple of 
free variables is $\prod_{\bar a \in S^+} \bar a$ (notice that we are using here 
the assumption that $\prod_{\bar a \in S^+} (\D,\bar a)$ is safe, 
i.e., that the free variables in $\prod_{\bar a \in S^+} \bar a$ do in fact appear in the atoms in 
$\D^{|S^+|}$). 
The CQ $\prod_{\bar a \in S^+} (\D,\bar a)$ is known as the {\em canonical} $\CQ$-explanation. 
We could then simply evaluate this canonical \CQ-explanation 
over $\D$ in order to meet the requirements of query by example. 
This, however, takes double exponential time since 
$\prod_{\bar a \in S^+} (\D,\bar a)$ itself is of exponential size. 
It is not known whether there are better algorithms 
for computing the 
result of {\em some} \CQ-explanation, but the results in this section suggest that this is 
unlikely. 

\medskip
\noindent
{\bf Size of CQ explanations and definitions.} 
It follows from the previous observations that $\CQ$-explanations are of at most
exponential size (by taking the canonical $\CQ$-explanation as witness). 
The same holds for $\CQ$-definitions. Interestingly, these bounds are optimal: 

  \begin{proposition} \label{prop:size-cq} {\em \cite{Willard10,CD15}} The following statements hold: 
\begin{enumerate} 
\item 
If there is a $\CQ$-explanation for 
$S^+$ and $S^-$ over $\D$, then there is a $\CQ$-explanation of at most exponential size; namely, 
$\prod_{\bar a \in S^+} (\D,\bar a)$. 
Similarly, for $\CQ$-definitions.  
\item There is a family $(\D_n,S^+_n,S^-_n)_{n \geq 0}$ of tuples of databases $\D_n$ and relations $S^+_n$ and $S^-_n$ over $\D_n$,  
such that (a) the combined size of $\D_n$, $S^+_n$, and $S^-_n$ is polynomial in $n$, 
(b) there is a $\CQ$-explanation for $S^+_n$ and $S^-_n$ over $\D_n$, and 
(c) the size of the smallest such $\CQ$-explanation is at least $2^{n}$. Similarly, for $\CQ$-definitions.
\end{enumerate}
\end{proposition}  

\medskip
\noindent
{\bf Sources of complexity.} 
The QBE test performs the following steps on input $(\D,S^+,S^-)$: 
(1) It computes $\prod_{\bar a \in S^+} (\D,\bar a)$, and (2) it checks whether 
$\prod_{\bar a \in S^+} (\D,\bar a)$ is unsafe or it is the case that 
$\prod_{\bar a \in S^+} (\D,\bar a)$ $\to (\D,\bar b)$ for some $\bar b \in S^-$. 
The definability test is equivalent, but the homomorphism test is then extended to each tuple over $\D$ but
outside $S^+$. Two sources of complexity are involved in these tests, each 
one of which incurs in one exponential blow up: (a) The construction of $\prod_{\bar a \in S^+} (\D,\bar a)$, and (b) 
the homomorphism tests  $\prod_{\bar a \in S^+} (\D,\bar a) \to (\D,\bar b)$. In order to alleviate the 
high complexity of the tests we thus propose relaxations of these two sources of complexity. 
The proposed relaxations are based on well-studied approximation 
notions with strong theoretical support. 
As such, they give rise to clean reformulations 
of the notions of 
$\CQ$-explanations/definitions. We start with the homomorphism test in the following section. 

\section{A relaxation of the homomorphism test} 
\label{sec:tw} 

We use an approximation technique for the homomorphism test  
based on the existential pebble game. This technique finds several 
applications in database theory \cite{KV95,DKV02} and can be shown to be equivalent
to the strong consistency tests for homomorphism approximation used in the area of constraint satisfaction \cite{KV-aaai}.
 The complexity of the (existential) pebble game is by now well-understood \cite{Grohe99,KP03}. 
We borrow several techniques used in such analysis to understand the 
complexity of our problems. We also prove some results on the complexity of such games that are of independent
interest. We define the existential pebble game below. 

\medskip
\noindent
{\bf The existential pebble game.} 
Let $k > 1$. The existential $k$-pebble game is played by the spoiler and the duplicator 
on pairs $(\D,\bar a)$ and $(\D',\bar b)$, where $\D$ and $\D'$ are databases over the same schema 
and $\bar a$ and $\bar b$ are $n$-ary ($n \geq 0$) tuples over $\D$ and $\D'$, respectively. 
The spoiler plays on $\D$ only, and the duplicator responds on $\D'$. 
In the first round the spoiler places his pebbles $\tt p_1,\dots,p_k$ on (not necessarily distinct) 
elements $c_1,\dots,c_k$ in $\D$, and the duplicator 
responds by placing his pebbles $\tt q_1,\dots,q_k$ on elements $d_1,\dots,d_k$ in $\D'$. In 
every further round, the spoiler removes one of his pebbles, say $\tt p_i$, for $1 \leq i \leq k$, 
and places it on an element of $\D$, 
and the duplicator responds by placing his corresponding pebble $\tt q_i$ on some element of 
$\D'$. The duplicator wins if he has a 
{\em winning strategy}, i.e., he can indefinitely continue playing the game in such way that 
at each round, if $c_1,\dots,c_k$ and $d_1,\dots,d_k$ are the elements covered by pebbles $\tt p_1,\dots,p_k$ and 
$\tt q_1,\dots,q_k$ on $\D$ and $\D'$, respectively, then 
$$\big((c_1,\dots,c_k,\bar a),(d_1,\dots,d_k,\bar b)\big)$$ 
is a {\em partial homomorphism} from $\D$ to $\D'$. 
Recall that this means that for every atom of the form $R(\bar c) \in \D$, where each element $c$ of $\bar c$ 
appears in $(c_1,\dots,c_k,\bar a)$, 
it is the case that $R(\bar d) \in \D'$, where $\bar d$ is the tuple that is obtained from $\bar c$ 
by replacing each element $c$ of 
$\bar c$ by its corresponding element $d$ in $(d_1,\dots,d_k,\bar b)$. 
If such winning strategy for the duplicator exists, we write $(\D,\bar a) \to_k (\D',\bar b)$. 

It is easy to see that the relations $\to_k$, for $k > 1$, 
provide an approximation of the notion of homomorphism in the following sense: 
$$\to \ \subsetneq \ \dots \ 
\subsetneq \ \to_{k+1} \ \subsetneq \ \to_k  \ \subsetneq \ \dots  \subsetneq \ \to_{2}.$$ 
Furthermore, these approximations are convenient from a complexity point of view: 
While checking for the existence of a homomorphism from $(\D,\bar a)$ to $(\D',\bar b)$ is \np-complete, 
checking for the existence of a winning strategy for the duplicator in the existential $k$-pebble game can be solved 
efficiently: 

\begin{proposition} \label{prop:games-poly} {\em \cite{KV95}} 
Fix $k > 1$. Checking if $(\D,\bar a) \to_k (\D',\bar b)$, given databases $\D$ and $\D'$ and 
$n$-ary tuples $\bar a$ and $\bar b$ over $\D$ and $\D'$, respectively, 
can be solved in polynomial time.  
\end{proposition} 

Furthermore, there is an interesting connection between the existential pebble game and the evaluation of 
CQs of bounded treewidth as established in the following proposition: 

\begin{proposition} \label{prop:games-tw} {\em \cite{AKV04}} 
Fix $k \geq 1$. Consider databases $\D$ and $\D'$ over the same schema and 
 $n$-ary tuples $\bar a$ and $\bar b$  over $\D$ and $\D'$, respectively. Then 
 $(\D,\bar a) \to_{k+1} (\D',\bar b)$ if and only if for each CQ $q(\bar x)$ in $\TW(k)$ such that 
$|\bar x| = n$ the following holds: $$\bar a \in q(\D) \ \ \Longrightarrow \ \ \bar b \in q(\D'),$$ 
or, equivalently, $(\D_q,\bar x) \to (\D,\bar a)$ implies $(\D_q,\bar x) \to (\D',\bar b)$, where as before $\D_q$ 
is the database that contains all the atoms of $q$.

Moreover, in case that $(\D,\bar a) \not\to_{k+1} (\D',\bar b)$ there exists an exponential size 
CQ $q(\bar x)$ in $\TW(k)$ such that $\bar a \in q(\D)$ but $\bar b \not\in q(\D')$.
\end{proposition} 

\medskip
\noindent
{\bf The relaxed test.} 
We study the following relaxed version of the QBE test for CQs 
that replaces the notion of homomoprhism $\to$ with its approximation 
$\to_k$, for a fixed $k > 1$: 

\begin{itemize} 
\item \underline{$k$-pebble QBE test for CQs:} Takes as input a 
database $\D$ and $n$-ary relations $S^+$ and $S^-$ over $\D$. It 
accepts if and only if:
\begin{enumerate} 
\item 
$\prod_{\bar a \in S^+} (\D,\bar a)$ is safe, and 
\item 
 $\prod_{\bar a \in S^+} (\D,\bar a) \not\to_k (\D,\bar b)$ for each tuple $\bar b \in S^-$. 
\end{enumerate} 
\end{itemize} 

Analogously, we define the $k$-pebble definability test for CQs. It immediately follows from 
the fact that the relation $\to_k$ can be decided in 
polynomial time (Proposition \ref{prop:games-poly}) that the $k$-pebble tests for CQs reduce the complexity of the general test from 
\conexp\ 
to \exptime. Later, in Section \ref{sec:comp}, we show that this is optimal.  

\subsection{A characterization of the $k$-pebble tests for CQs}
 
Using Proposition \ref{prop:games-tw} we can now establish the theoretical 
meaningfulness of the relaxed tests: They admit a clean characterization in terms of the 
CQs of bounded treewidth. 
In fact, recall that the QBE (resp., definability) 
test for CQs precisely characterizes the existence of \CQ-explanations (resp., \CQ-definitions). As we 
show next, their relaxed versions based on the existential $(k+1)$-pebble game preserve these characterizations
up to treewidth $k$:   

\begin{theorem} \label{theo:games} 
Fix $k \geq 1$. 
Consider a database $\D$ and $n$-ary 
relations $S^+$ and $S^-$ over $\D$. 
\begin{enumerate} 
\item 
There is a $\TW(k)$-explanation for $S^+$ and $S^-$ over $\D$ if and only if
the $(k+1)$-pebble QBE test for CQs accepts $\D$, $S^+$, and $S^-$.
\item
 There is a $\TW(k)$-definition  
for $S^+$ over $\D$ if and only if
the $(k+1)$-pebble definability test for CQs 
accepts $\D$ and $S^+$.
\end{enumerate}   
\end{theorem}  

\begin{proof} 
We concentrate on explanations (the proof for definitions is analogous). 
From left to right, assume for the sake of contradiction 
that $q$ is a $\TW(k)$-explanation for $S^+$ and $S^-$ over $\D$, yet the $(k+1)$-pebble 
QBE test for CQs fails over $\D$, $S^+$, and $S^-$. Since there is a $\TW(k)$-explanation for $S^+$ and $S^-$ over $\D$, 
we have from 
Proposition \ref{prop:folklore} that $\prod_{\bar a \in S^+} (\D,\bar a)$ is safe.  
Therefore, it must be the case that $\prod_{\bar a \in S^+} (\D,\bar a) \to_{k+1} (\D,\bar b)$ for some 
$\bar b \in S^-$. Since $S^+ \subseteq q(\D)$ it is the case that $\bar a \in q(\D)$ for each $\bar a \in S^+$. That is, 
$(\D_q,\bar x) \to (\D,\bar a)$ for each $\bar a \in S^+$. Due to basic properties of direct products, 
this implies that $(\D_q,\bar x) \to \prod_{\bar a \in S^+} (\D,\bar a)$. From Proposition \ref{prop:games-tw} 
we conclude that $(\D_q,\bar x) \to (\D,\bar b)$, i.e., $\bar b \in q(\D)$. This is a contradiction since 
$\bar b \in S^-$ and $q(\D) \cap S^- = \emptyset$. 

From right to left, assume that the $(k+1)$-pebble QBE test for CQs accepts 
$\D$, $S^+$, and $S^-$, i.e., $\prod_{\bar a \in S^+} (\D,\bar a)$ is safe and 
for every tuple $\bar b \in S^-$ it is the case that 
$\prod_{\bar a \in S^+} (\D,\bar a) \not\to_{k+1} (\D,\bar b)$. Since $\prod_{\bar a \in S^+} (\D,\bar a)$ is safe we can apply Proposition \ref{prop:games-tw}, which tells us that for each $\bar b \in S^-$ 
there is a CQ $q_{\bar b}(\bar x)$ such that $(\D_{q_{\bar b}},\bar x) \to \prod_{\bar a \in S^+} (\D,\bar a)$ but
$(\D_{q_{\bar b}},\bar x) \not\to (\D,\bar b)$. Suppose first that $S^- \neq \emptyset$ 
and let: $$q(\bar x) \ := \ \bigwedge_{\bar b \in S^-} q_{\bar b}(\bar x).$$ 
It is easy to see that $q(\bar x)$ is well-defined (since $S^-$ is nonempty) and 
can be expressed as a CQ in $\TW(k)$. For the latter we simply use fresh existentially quantified variables for 
each CQ $q_{\bar b}$ such that $\bar b \in S^-$ and then move all existentially quantified variables in 
$\bigwedge_{\bar b \in S^-} q_{\bar b}(\bar x)$ 
to the front. We now prove that $q(\bar x)$ is a $\TW(k)$-explanation for $S^+$ and $S^-$ over $\D$. 
It easily follows that
$(\D_{q},\bar x) \to \prod_{\bar a \in S^+} (\D,\bar a)$ from the fact that 
$(\D_{q_{\bar b}},\bar x) \to \prod_{\bar a \in S^+} (\D,\bar a)$ for each $\bar b \in S^-$. 
But then $(\D_{q},\bar x) \to (\D,\bar a)$ 
for each $\bar a \in S^+$. This means that $\bar a \in q(\D)$ for each $\bar a \in S^+$, i.e., 
$S^+ \subseteq q(\D)$. 
Assume now for the sake of contradiction that $q(\D) \cap S^- \neq \emptyset$, that is, there is a tuple $\bar b \in q(\D) \cap S^-$.  
Then $(\D_q,\bar x) \to (\D,\bar b)$, which implies that $(\D_{q_{\bar b}},\bar x) \to (\D,\bar b)$. This is a contradiction. 
The case when $S^- = \emptyset$ can be proved 
using similar techniques. 
\end{proof} 

\subsection{The complexity of the $k$-pebble tests for CQs} 
\label{sec:comp} 

As mentioned before, the $k$-pebble tests for CQs can be evaluated 
in exponential time. We show here that such bounds are also optimal: 

\begin{theorem} \label{theo:games-tests} 
Deciding whether the $k$-pebble QBE test for CQs accepts $(\D,S^+,S^-)$ 
is \exptime-complete for each $k > 1$. Similarly, for the $k$-pebble definability 
test for CQs. 
This holds even if restricted to the case when $S^+$ and $S^-$ are unary relations. 
\end{theorem}

As a corollary to Theorems \ref{theo:games} and \ref{theo:games-tests}, 
we obtain the following interesting result: 

\begin{corollary} \label{coro:completeness} 
The problems {\sc $\TW(k)$-query-by-example} and 
{\sc $\TW(k)$-definability} are \exptime-complete for each fixed $k \geq 1$. 
This holds even if restricted to the case when the relations to be explained/defined are 
unary. 
\end{corollary} 

We now provide a brief outline of the main ideas used for proving the lower bounds in Theorem \ref{theo:games-tests}. 
Let us first notice that in the case of the general QBE/definability tests for CQs, a \conexp\ lower bound is 
obtained in \cite{CD15} as follows: 
\begin{enumerate} 
\item It is first shown that the following 
{\em product homomorphism problem} (PHP) is \nexp-hard: 
Given databases 
$\D_1,\dots,\D_m$ and $\D$, is it the case that $\prod_{1 \leq i \leq m} \D_i \to \D$?  
\item 
It is then shown that there is an easy polynomial-time 
reduction from PHP to the problem of checking whether 
the QBE/definability test fails on its input. 
\end{enumerate} 

The ideas used for proving (2) can be easily adapted to show that there is a polynomial-time reduction from the following 
{\em relaxed} version of PHP to the problem of checking whether the $k$-pebble QBE/definability test fails on its input: 

\begin{center}
\fbox{\begin{tabular}{ll}
\small{PROBLEM} : & {\sc $k$-pebble PHP} (for $k > 1$)
\\{\small INPUT} : & Databases $\D_1,\dots,\D_m$ and $\D$ over the same schema 
\\
{\small QUESTION} : &  Is it the case that $\prod_{1 \leq i \leq m} \D_i \to_{k} \D$?
\end{tabular}}
\end{center}

We establish that this relaxed version of PHP is \exptime-complete for each fixed 
$k > 1$: 

\begin{theorem} \label{theo:ii} 
The problem {\sc $k$-pebble PHP} is \exptime-complete for each fixed $k > 1$. 
\end{theorem} 

To prove this result, we exploit techniques from \cite{KP03,Grohe99} that study the complexity of pebble games. 
In particular, it is shown in \cite{KP03} that for each fixed $k > 1$, 
checking whether $\D \rightarrow_k \D'$ is {\ptime}-complete. 
The proof uses an involved 
reduction from the {\em monotone circuit value problem}, that is, given a monotone circuit $C$, 
it constructs two databases $\D_C$ and $\D_C'$ such that the value of $C$ is $1$ if and only if 
$\D_C\rightarrow_k \D_C'$.

In our case, to show that {\sc $k$-pebble PHP} is {\exptime}-hard for each fixed $k > 1$, 
we reduce from the following well-known {\exptime}-complete problem: 
Given an alternating Turing machine $M$ and a positive integer $n$, decide whether $M$ accepts the 
empty tape using $n$ space. 
The latter problem can be easily recast as a circuit value problem: We can construct 
a circuit $C_{M,n}$ such that the value of $C_{M,n}$ is 1 if and only if $M$ 
accepts the empty tape using $n$ space. 
The main idea of our reduction is to construct databases $\D_1,\dots,\D_m$ and $\D$, 
given $M$ and $n$, such that: $$\prod_{1\leq i\leq m}\D_i \rightarrow_k \D \ \ \Longleftrightarrow \ \ 
\D_{C_{M,n}} \, \rightarrow_k \, \D_{C_{M,n}}',$$ 
where $\D_{C_{M,n}}$ and $\D_{C_{M,n}}'$ are defined as in \cite{KP03}. 

A natural approach then is to construct $\D_1,\dots,\D_m,\D$ such that $\prod_{1\leq i\leq m}\D_i$ and $\D$ roughly 
coincide with $\D_{C_{M,n}}$ and $\D_{C_{M,n}}'$. However, there is a problem with this: 
the databases $\D_{C_{M,n}}$ and $\D_{C_{M,n}}'$ closely resemble the circuit $C_{M,n}$, 
but the size of $C_{M,n}$ is exponential in $|M|$ and $n$, and so are the sizes of $\D_{C_{M,n}}$ and $\D_{C_{M,n}}'$. 
Although it is possible to codify the exponential size database $\D_{C_{M,n}}$ using 
a product of polynomial size databases $\D_1,\dots,\D_m$, 
we cannot do the same with the exponential size $\D_{C_{M,n}}'$ using $\D$ only. 
To overcome this, we need to extend the techniques in \cite{KP03} and show that the 
complexity of the existential $k$-pebble game is \ptime-complete even over a {\em fixed template}: 

\begin{lemma} \label{lemma:template} 
For each fixed $k>1$, there is a database $\D_k$ that only depends on $k$, such that the following problem is 
{\ptime}-complete: Given a database $\D$, decide whether $\D \rightarrow_k \D_k$.
\end{lemma}

To prove this, we again use a reduction from the circuit value problem that given a circuit $C$ constructs 
a database $\tilde \D_C$ such that $C$ takes value 1 if and only if $\tilde \D_C \to_k \D_k$.  
We then use the following idea to prove that {\sc $k$-pebble PHP} is \exptime-complete: Given 
$M$ and $n$, we construct in polynomial time databases $\D_1,\dots,\D_m$ and $\D$ such that 
$\prod_{1\leq i\leq m}\D_i$ and $\D$ roughly coincide with $\tilde \D_{C_{M,n}}$ and $\D_k$, respectively. 
It then follows that: $$\prod_{1\leq i\leq m}\D_i \, \rightarrow_k \, \D 
\ \ \Longleftrightarrow \ \ \tilde \D_{C_{M,n}} \, \rightarrow_k \, \D_k \ \ \Longleftrightarrow \ \ 
\text{$M$ accepts the 
empty tape using $n$ space.}$$  

\subsection{Evaluating the result of $\TW(k)$-explanations}

Recall that computing the result of $\CQ$-explanations might require double exponential time. For 
$\TW(k)$-explanations, instead, we can do this in single exponential time. 

\begin{theorem} \label{theo:eval} Fix $k \geq 1$. 
There is a single exponential time algorithm that, given a database $\D$ and $n$-ary relations $S^+$ and $S^-$ over $\D$, does 
the following: 
\begin{enumerate}
\item It checks whether there is a $\TW(k)$-explanation for $S^+$ and $S^-$ over $\D$, and
\item if the latter holds, it computes the evaluation $q(\D)$ of one such $\TW(k)$-explanation $q$.  
\end{enumerate}  
\end{theorem} 

\begin{proof} 
We first check in exponential time the existence of one such $\TW(k)$-explanation for $S^+$ and $S^-$ over $\D$ 
using the $(k+1)$-pebble QBE test for CQs. If such $\TW(k)$-explanation exists, we compute in 
exponential time the set $S^e$ of all $n$-ary tuples $\bar b$ over $\D$ such that $\prod_{\bar a \in S^+} (\D,\bar a) 
\to_{k+1} (\D,\bar b)$. Notice, in particular, that $S^+ \subseteq S^e$ and $S^e \cap S^- = \emptyset$. Moreover, it can be shown that $S^e = q(\D)$ for some $\TW(k)$-explanation $q$ for $S^+$ and $S^-$ over $\D$. 
\end{proof} 

Notably, the previous result computes the result of a $\TW(k)$-explanation $q$ for $S^+$ and $S^-$ over $\D$ 
without explicitly computing $q$. One might wonder whether it is possible to also include $q$ in 
the output of the algorithm. The answer is negative, and the reason is that 
$\TW(k)$-explanations/definitions can be double exponentially large 
in the worst case: 

\begin{proposition} {\em \label{prop:size-ub}}   
Fix $k \geq 1$. The following holds:
\begin{enumerate} 
\item 
Assume that there is a $\TW(k)$-explanation for $S^+$ and $S^-$ over $\D$. Then there is one such $\TW(k)$-explanation
of at most double exponential size. 
\item 
There is a family $(\D_n,S^+_n,S^-_n)_{n \geq 0}$ of tuples of databases $\D_n$ and relations $S^+_n$ and $S^-_n$ over $\D_n$, 
such that (a) the combined size of $\D_n$, $S^+_n$, and $S^-_n$ is polynomial in $n$, (b) there is a 
$\TW(k)$-explanation for $S^+_n$ and $S^-_n$ over $\D_n$, and 
(c) the size of the smallest such  
$\TW(k)$-explanation is at least $2^{2^n}$. 
\end{enumerate} 
The same holds for $\TW(k)$-definitions.
\end{proposition} 

\begin{proof}
From the proof of Theorem \ref{theo:games}, whenever there is a $\TW(k)$-explanation for $S^+$ and $S^-$ over $\D$ 
this can be assumed to be the CQ 
$q = \bigwedge_{\bar b \in S^-} q_{\bar b}(\bar x)$. 
From Proposition \ref{prop:games-tw}, each such $q_{\bar b}$ is of exponential size in the combined size of $\prod_{\bar a \in S^+} (\D,\bar a)$ and $(\D,\bar b)$, i.e., double exponential 
in the size of $\D$, $S^+$ and $S^-$. Thus, the size of 
$q$ is at most double exponential in that of $\D$, $S^+$ and $S^-$.  
The lower bound follows by inspection of the proof of Theorem \ref{theo:games-tests}.  
\end{proof} 

Notice that this establishes a difference with \CQ-explanations/definitions, 
which are at most of exponential size (see Proposition \ref{prop:size-cq}). 

\section{Desynchronizing the direct product}
\label{sec:des} 

We now look at the other source of complexity for the QBE and definability tests for CQs: The construction 
of the direct product $\prod_{\bar a \in S^+} (\D,\bar a)$. It is a priori not obvious how to define 
reasonable approximations of this construction with a meaningful theoretical interpretation. As a first step in 
this direction, we look at a simple idea that has been applied in the study of $\CQ$-definability: We ``desynchronize'' this 
direct product and consider each tuple $\bar a \in S^+$ in isolation. This leads to the following relaxed test:  

\begin{itemize} 
\item \underline{Desynchronized QBE test for CQs:} Takes as input a 
database $\D$ and $n$-ary relations $S^+,S^-$ over $\D$. It 
accepts iff for each $\bar a \in S^+$ and $\bar b \in S^-$ it is the case that
$(\D,\bar a) \not\to (\D,\bar b)$.  
\end{itemize} 

Similarly, we define the desynchronized definability test for CQs. 
Notice that, unlike the previous tests we have presented in the paper, the 
desynchronized tests do not require any safeness condition (for reasons we explain below). 

It follows from \cite{ANS13} that these tests capture the notion of explanations/definitions 
for the class of {\em unions} of CQs (UCQs). Recall that a UCQ is a formula $Q$ of the form 
$\bigvee_{1 \leq i \leq m} q_i(\bar x)$, where the $q_i(\bar x)$'s are CQs over 
the same schema. The evaluation $Q(\D)$ of $Q$ over database $\D$ corresponds to $\bigcup_{1 \leq i \leq m} q_i(\D)$. 
We denote by $\UCQ$ the class of UCQs. 
We then obtain the following: 

\begin{theorem}[implicit in \cite{ANS13}] \label{theo:des}  
Consider a database $\D$ and $n$-ary 
relations $S^+$ and $S^-$ over $\D$. 
There is a $\UCQ$-explanation for $S^+$ and $S^-$ over $\D$ if and only if
the desynchronized QBE test for CQs 
accepts $\D$, $S^+$, and $S^-$.
Similarly, for the $\UCQ$-definitions of $S^+$ and the desynchronized definability test for CQs. 
\end{theorem}  

In this case, the {\em canonical} $\UCQ$-explanation/definition corresponds to 
$Q = \bigcup_{\bar a \in S^+} (\D,\bar a)$.  This explains why no safeness condition is required on the desynchronized tests,
as each pair of the form $(\D,\bar a)$, for $\bar a \in S^+$, is safe by definition. 
Notice that $Q$ consists of polynomially many CQs of polynomial size. Its evaluation $Q(\D)$ over a database 
$\D$ can thus be computed in single exponential time (as opposed to the 
double exponential time needed to evaluate the canonical \CQ-explanation $\prod_{\bar a \in S^+} (\D,\bar a)$).  

It is easy to see that the desynchronization of the direct product reduces the complexity 
of the general tests from \conexp\ to \coNP. It follows from \cite{ANS13} that this bound is optimal. 
As a corollary to Theorem \ref{theo:des} we thus obtain that QBE/definability for UCQs are \coNP-complete:  

\begin{proposition} \label{prop:des-comp}  {\em \cite{ANS13}} 
The following statements hold:  
\begin{enumerate} 
\item 
Deciding whether the desynchronized QBE test for CQs accepts $(\D,S^+,S^-)$ 
is \coNP-complete. Similarly, for the desynchronized definability test for CQs. 
\item 
{\sc $\UCQ$-query-by-example} and 
{\sc $\UCQ$-definability} are \coNP-complete.
\end{enumerate}
\end{proposition}

\subsection{Combining both relaxations} 

By combining both relaxations (replacing homomorphism tests with relations $\to_k$, for $k > 1$, and desynchronizing direct products) 
we obtain the {\em desynchronized $k$-pebble QBE (resp., definability) test for CQs}. Its definition coincides with that of the 
desynchronized QBE (resp., definability) test for CQs given above, save that now the homomorphism test $(\D,\bar a) \to (\D,\bar b)$ 
is replaced by $(\D,\bar a) \to_k (\D,\bar b)$. As is to be expected from the previous charaterizations, this test captures definability by the class 
of UCQs of bounded treewidth. 
Formally, let $\UTW(k)$ be the class of unions of CQs in $\TW(k)$ (for $k \geq 1$). Then:  

\begin{theorem} \label{theo:comb}  
Fix $k \geq 1$. Consider a database $\D$ and $n$-ary 
relations $S^+$ and $S^-$ over $\D$. 
There is a $\UTW(k)$-explanation for $S^+$ and $S^-$ over $\D$ if and only if
the desynchronized $(k+1)$-pebble QBE test for CQs 
accepts $\D$, $S^+$, and $S^-$.
Similarly, for the $\UTW(k)$-definitions of $S^+$ and the desynchronized $(k+1)$-pebble 
definability test for CQs.
\end{theorem}  

Furthermore, in case there is a $\UTW(k)$-explanation for $S^+$ and $S^-$ over $\D$ (resp., 
a $\UTW(k)$-definition of $S^+$ over $\D$), then there is one such explanation/definition 
given by a union of polynomially many CQs in $\TW(k)$, each one of which is of at most exponential size.  

Interestingly, the combination of both relaxations yields tractability for the QBE test. In contrast, 
the definability test remains \coNP-complete. The difference lies on the fact that the QBE test 
only needs to perform a 
polynomial number of tests of the form 
$(\D,\bar a) \to_{k} (\D,\bar b)$ for each $\bar a \in S^+$ (one for each tuple $\bar b \in S^-$), 
while the definability test needs to perform 
exponentially many such tests (one for each tuple $\bar b$ outside $S^+$). Then:  

\begin{proposition} \label{prop:des-comb-comp} 
The following statements hold: 
\begin{enumerate} 
\item 
Deciding whether the desynchronized $k$-pebble 
QBE test for CQs accepts $(\D,S^+,S^-)$ 
can be solved in polynomial time for each fixed $k > 1$. As a consequence, 
{\sc $\UTW(k)$-query-by-example} is in polynomial time for each fixed $k \geq 1$.
\item 
If a $\UTW(k)$-explanation for $S^+$ and $S^-$ over $\D$ exists, 
we can compute the evaluation $Q(\D)$ of one such explanation $Q$ in exponential time. 
\item 
Deciding whether the desynchronized $k$-pebble 
definability test for CQs  accepts $(\D,S^+)$ 
is \coNP-complete for each fixed $k > 1$. As a consequence, 
{\sc $\UTW(k)$-definability} is \coNP-complete for each $k \geq 1$. 
\end{enumerate}
\end{proposition}

\section{Conjunctive regular path queries}
\label{sec:ext} 

We now switch to study the QBE and definability problems in the context of graph databases. 
Let $\Sigma$ be a finite alphabet. Recall that a {\em graph
database} $\G = (V,E)$ over $\Sigma$ consists of a finite set $V$ of nodes and a set 
$E \subseteq V \times \Sigma \times V$ of directed edges labeled in $\Sigma$ (i.e., $(v,a,v') \in E$ represents the fact that 
there is an $a$-labeled edge from node $v$ to node $v'$ in $\G$). 
A {\em path} in $\G$ is a sequence 
$$\eta \ = \ v_0 a_1 v_1 a_2 v_2 \dots v_{k-1} a_{k} v_k, \ \ \ \ \text{for $k \geq 0$,}$$   
such that $(v_{i-1},a_i,v_i) \in E$ for each $1 \leq i \leq k$. 
The {\em label} of $\eta$, denoted $\lab(\eta)$, is the
word $a_1 a_2 \dots a_{k}$ in $\Sigma^*$. 
Notice that $v$ is a
path for each node $v \in V$. 
The label of such
path is the empty word $\varepsilon$. 

The basic navigational mechanism for querying graph databases is the class of
{\em regular path queries}, or RPQs (see, e.g., \cite{Woo, Bar13}).  
An RPQ $L$ over alphabet $\Sigma$ is a regular
expression over 
$\Sigma$. The {\em evaluation} $L(\G)$ of $L$ over
graph database $\G$ consists of those pairs $(v,v')$ of nodes in $\G$ such that there
is a path $\eta$ in $\G$ from $v$ to $v'$ whose label 
$\lab(\eta)$ satisfies $L$. The analogue of CQs in the context of graph databases is the class
of {\em conjunctive} RPQs, or CRPQs \cite{CGLV00}. 
Formally, a CRPQ $\gamma$ over $\Sigma$
is an expression of the form:  
$$ 
\exists \bar z (L_1(x_1,y_1) \wedge \dots \wedge L_m(x_m,y_m)),
$$
where each $L_i$ is a RPQ over $\Sigma$, for $1 \leq i \leq m$, and $\bar z$ is a tuple of
variables among $\{x_1,y_1,\dots,x_m,y_m\}$. We write $\gamma(\bar x)$ to
denote
that $\bar x$ is the tuple of free
variables of $\gamma$. A homomorphism from $\gamma$ to the graph database 
$\G$ is a mapping
$h$ from $\{x_1,y_1,\dots,x_m,y_m\}$ to the nodes of $\G$, such that
$(h(x_i),h(y_i)) \in L_i(\G)$ for each $1 \leq i \leq m$. 
The evaluation $\gamma(\G)$ of $\gamma(\bar x)$ over $\G$ is 
the set of tuples $h(\bar x)$ such that $h$ a homomorphism from $\gamma$
to $\G$.  We denote the class of CRPQs by $\CRPQ$. 

\subsection{The QBE and definability tests for CRPQs}

We present QBE/definability tests for CRPQs in the same spirit than the tests for CQs, save that we now 
use a  
notion of {\em strong homomorphism} from a product $\prod_{1 \leq i \leq n} \G_i$ 
of directed graphs to a single directed graph $\G$. This notion preserves, in a precise sense defined below, 
the languages defined by pairs of nodes in $\prod_{1 \leq i \leq n} \G_i$.
Interestingly, these tests yield a {\conexp} upper bound for the QBE/definability problems for CRPQs, which 
improves the \expspace\ upper bound from \cite{ANS13}. 
In conclusion, QBE/definability for CRPQs is no more difficult than for CQs.

We start with some notation. 
Let $v$ and $v'$ be nodes in a graph database $\G$. We define the following language in 
$\Sigma^*$: 
$$L_{v,v'}^\G \ := \ \{\lab(\eta)\mid \text{$\eta$ is a path in $\G$ from $v$ to $v'$}\}.$$ 
Moreover, if $\G_1 = (V_1,E_1)$ and $\G_2 = (V_2,E_2)$ are graph databases over $\Sigma$, their direct product 
$\G_1 \otimes \G_2$ is the graph database $(V,E)$ such that $V = V_1 \times V_2$ and there is an $a$-labeled edge in $E$ from 
node $(v_1,v_2)$ to node $(v_1',v_2')$ if and only if $(v_1,a,v_2) \in E_1$ and $(v_1',a,v_2') \in E_2$.  

Let then $\G_1,\dots,\G_n$ and $\G$ be graph databases over $\Sigma$. 
A {\em strong homomorphism} from 
$\prod_{1\leq i\leq n}\G_i$ to $\G$ is a mapping $h$ from the nodes of 
$\prod_{1\leq i\leq n}\G_i$ to the nodes of $\G$ such that for each pair  
$\bar v=(v_1,\dots,v_n)$ and $\bar v'=(v'_1,\dots,v'_n)$ of nodes in $\prod_{1\leq i\leq n}\G_i$, it is the case that: 
$$L_{v_i,v_i'}^{\G_i} \ \subseteq \ L_{h(\bar v), h(\bar v')}^\G, \ \ \ \text{ for some coordinate $i$ with $1 \leq i \leq n$.}$$  
We write $\prod_{1\leq i\leq n}\G_i\Rightarrow \G$ when there is a strong homomorphism $h$ from 
$\prod_{1\leq i\leq n}\G_i$ to $\G$. 
Note that in this case, $h$ must also be a (usual) homomorphism from 
$\prod_{1\leq i\leq n}\G_i$ to $\G$, i.e., $\prod_{1\leq i\leq n}\G_i\Rightarrow \G$ implies 
$\prod_{1\leq i\leq n}\G_i\rightarrow \G$. The next example shows that the converse does not hold in general:

\begin{example}
\label{exa:strong-homo}
Let $\vec C_n$ be the directed cycle of length $n$ over  $\{1,2,\dots,n\}$. 
We assume $\vec C_n$ to be represented as a graph database over the unary alphabet $\Sigma = \{a\}$. 
We then have that $\vec C_2\otimes \vec C_3 \rightarrow \vec C_6$, since $\vec C_2\otimes \vec C_3$ is isomorphic to $\vec C_6$ 
as shown below (we omit the labels):
\begin{figure}[h!]
\centering
\includegraphics[scale=0.4]{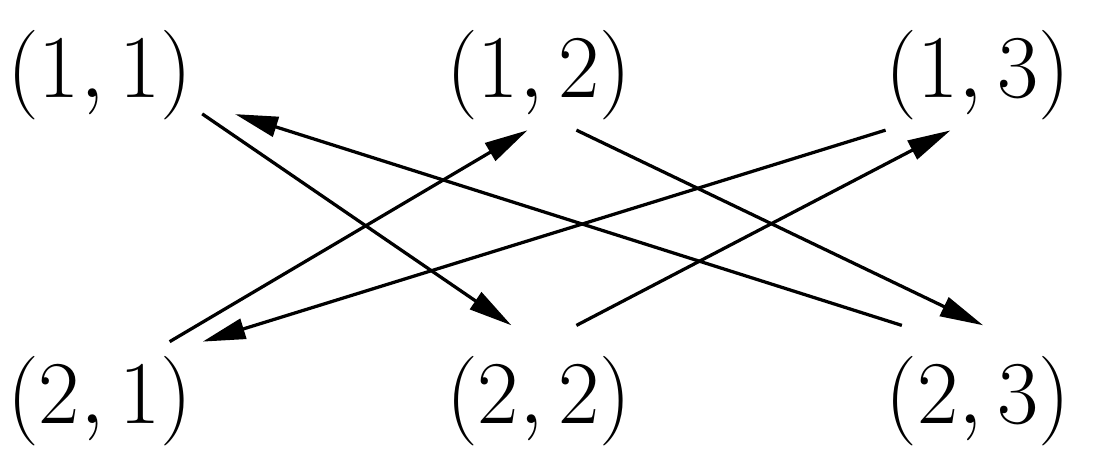}
\label{fig:reduction}
\end{figure}

\noindent
On the other hand, 
$\vec C_2\otimes \vec C_3 \not\Rightarrow \vec C_6$. To see this, take e.g. the homomorphism $h$ defined as 
$$\{(1,1)\mapsto 1, \, (2,2)\mapsto 2, \, (1,3)\mapsto 3, \, (2,1)\mapsto 4, \, (1,2)\mapsto 5, \, 
(2,3)\mapsto 6\}.$$ This is not a strong homomorphism as witnessed by the pair $(1,1)$ and $(2,2)$. 
Indeed, we have that: $$\big(\text{$h(1,1)=1$} \ \text{ and } \ \text{$h(2,2)=2$}\big) \ \text{ but } \ 
\big(\text{$L_{1,2}^{\vec C_2}\not \subseteq L_{1,2}^{\vec C_6}$} 
\ \text{ and } \ \text{$L_{1,2}^{\vec C_3}\not \subseteq L_{1,2}^{\vec C_6}$.}\big)$$ 
The reason is that $aaa \in L_{1,2}^{\vec C_2}$, $aaaa 
\in L_{1,2}^{\vec C_3}$, 
but none of these words is in $L_{1,2}^{\vec C_6}$. The same holds for any homomorphism $h : \vec C_2\otimes \vec C_3 \to 
\vec C_6$. \qed
\end{example}

If $(\G_1,\bar a_1),\dots, (\G_n,\bar a_n)$ and $(\G,\bar b)$ are graph databases with distinguished tuple of elements, then 
we write $\prod_{1\leq i\leq n}(\G_i,\bar a_i)\Rightarrow (\G, \bar b)$ if there is a strong homomorphism $h$ from $\prod_{1\leq i\leq n}\G_i$ to $\G$ 
such that $h(\bar a_1\otimes\cdots\otimes\bar a_n)=\bar b$. Next we present our tests for CRPQs:

\begin{itemize} 
\item \underline{QBE test for CRPQs:} Takes as input a 
graph database $\G$ and $n$-ary relations $S^+$ and $S^-$ over $\G$. It 
accepts if and only if $\prod_{\bar a \in S^+} (\G,\bar a) \not\Rightarrow (\G,\bar b)$ for each tuple $\bar b \in S^-$. 
\item \underline{Definability test for CRPQs:} Takes as input a 
graph database $\G$ and an $n$-ary relation $S^+$ over $\G$. It 
accepts if and only if $\prod_{\bar a \in S^+} (\G,\bar a) \not\Rightarrow (\G,\bar b)$ for each $n$-ary tuple $\bar b\notin S^+$.
\end{itemize} 

As it turns out, our tests characterize the non-existence of $\CRPQ$-explanations/definitions.
(Notice that unlike Proposition \ref{prop:folklore}, we need no safety conditions on QBE/definability 
tests for CRPQs 
for this characterization to hold). 

\begin{theorem}
\label{theo:tests-crpq}
The following hold:
\begin{enumerate}
\item Let $\G$ be a database and $S^+,S^-$ relations over $\G$. There is a $\CRPQ$-explanation for 
$S^+$ and $S^-$ 
over $\G$ if and only if the QBE test for CRPQs accepts $\G$, $S^+$, and $S^-$.
\item Let $\G$ be a database and $S^+$ a relation over $\G$. There is a \CRPQ-definition for $S^+$
over $\G$ if and only if the definability test for CRPQs accepts $\G$ and $S^+$.
\end{enumerate} 
\end{theorem}

Since containment of regular languages can be checked in polynomial space \cite{SM73}, 
it is straightforward to check that both tests can be carried out in {\conexp}. Thus: 

\begin{theorem}
\label{theo:upper-crpq}
{\sc \CRPQ-query-by-example} and {\sc \CRPQ-definibility} are in {\conexp}.
\end{theorem} 

Whether these problems are complete for {\conexp} is left as an open question. 

\medskip
\noindent
{\bf CRPQ vs UCQ explanations.} 
It is easy to see that if there is a $\CRPQ$-explanation for $S^+$ and $S^-$ over $\G$, then there is also a $\UCQ$-explanation 
\cite{ANS13}. 
One may wonder then if QBE for CRPQs and UCQs coincide. 
If this was the case, we would directly obtain a {\coNP} upper bound for {\sc $\CRPQ$-query-by-example} 
from Proposition
\ref{prop:des-comp} (which establishes that {\sc \UCQ-query-by-example} is in \coNP). 
The next example shows that this is not the case: 

\begin{figure}
\centering
\includegraphics[scale=0.75]{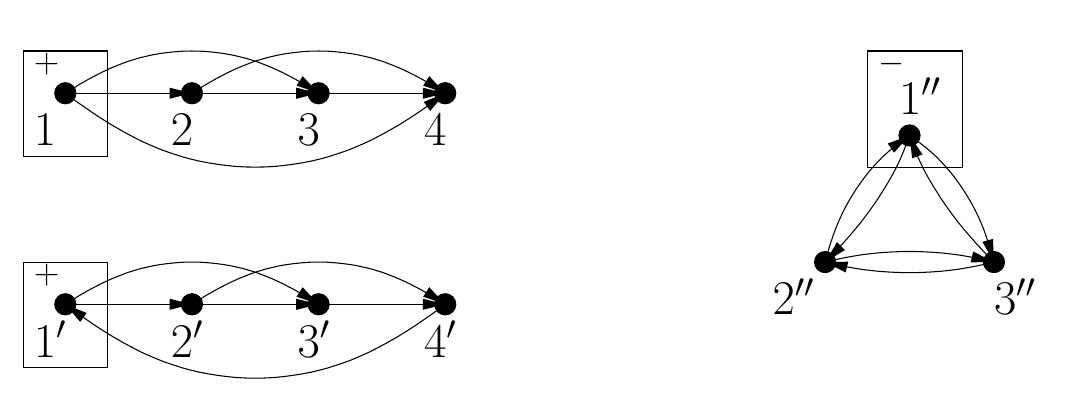}
\caption{The graph database $\G$ from Example \ref{ex:converse}.}
\label{fig:reduction}
\end{figure}

\begin{example} \label{ex:converse} 
Consider the graph database $\G$ over $\Sigma = \{a\}$ given by the three connected components depicted in Figure \ref{fig:reduction}
(we omit the labels). 
Let $S^+=\{1,1'\}$ and $S^-=\{1''\}$. Clearly, $(\G,1)\not\rightarrow (\G,1'')$ and $(\G,1')\not\rightarrow (\G,1'')$, 
since the underlying graph of each component on the left-hand side is a clique of size $4$, 
while the one on the right-hand side is a clique of size $3$. 
It follows that there is a \UCQ-explanation for $S^+$ and $S^-$ over $\G$. 
On the other hand, a straightforward construction shows that $(\G,1)\otimes(\G,1')\Rightarrow (\G,1'')$. 
The intuition is that, since $(4',1')$ and $(1,4)$ have opposite direction, they do not synchronize in the product and, thus, the product does not contain a
clique of size $4$. We conclude that there is no $\CRPQ$-explanation for $S^+$ and $S^-$ over $\G$. \qed
\end{example}
  
 \subsection{Relaxing the QBE and definability tests for CRPQs}
 
 In this section, we develop relaxations of the tests for CRPQs based on the ones we 
 studied for CQs in the previous sections. 
 Let us start by observing that desynchronizing the direct product 
 trivializes the problem in this case: In fact, as expected 
 the {\em desynchronized QBE/definability tests for CRPQs} characterize QBE/definability for the 
 {\em unions} of CRPQs (\UCRPQ). It is known, on the other hand, that QBE/definability for \UCRPQ\ and \UCQ\ coincide 
 \cite{ANS13}. The results then follow directly from the ones obtained in Section \ref{sec:des} for UCQs. In particular, 
 {\sc $\UCRPQ$-query-by-example} and {\sc $\UCRPQ$-definability} are \coNP-complete.  
 
 We thus concentrate on the most interesting case, which is 
 the relaxation of the homomorphism tests. 
 In order to approximate the strong homomorphism test, 
 we consider a variant of the existential pebble game. Fix $k > 1$. 
 Let $(\G_1,\bar a_1),\dots,(\G_n,\bar a_n)$ and $(\G,\bar b)$ be graph databases over $\Sigma$ 
 with distinguished tuples of elements. We define $\bar a := \bar a_1\otimes\cdots\otimes \bar a_n$.
 The {\em strong existential $k$-pebble game} 
 on $\prod_{1\leq i\leq n}(\G_i,\bar a_i)$ and $(\G,\bar b)$ is played as the 
 existential $k$-pebble game on $\prod_{1\leq i\leq n}(\G_i,\bar a_i)$ and $(\G,\bar b)$, but now, at each round, 
 if $c_1,\dots,c_k$ and $d_1,\dots,d_k$ are the elements covered by pebbles on $\prod_{1\leq i\leq n}\G_i$ and $\G$, respectively, then 
 the duplicator needs to ensure that $((c_1,\dots,c_k,\bar a),(d_1,\dots,d_k,\bar b))$ is a {\em strong partial homomorphism} 
 from $\prod_{1\leq i\leq n}\G_i$ and $\G$. 
 This means that for every pair $\bar v=(v_1,\dots,v_n)$ and $\bar v'=(v_1',\dots,v'_n)$ of nodes in 
 $\prod_{1\leq i\leq n}\G_i$ 
 that appear in $(c_1,\dots,c_k,\bar a)$, if $u$ and $u'$ are the elements in 
 $(d_1,\dots,d_k,\bar b)$ that correspond to $\bar v$ and $\bar v'$, respectively, then: $$L^{\G_i}_{v_i,v_i'} \ \subseteq \  
 L^{\G}_{u,u'}, \ \ \ \text{ for some coordinate $i$ with $1 \leq i \leq n$.}$$  
 We write $\prod_{1\leq i\leq n}(\G_i,\bar a_i)\Rightarrow_k(\G,\bar b)$ if the 
 duplicator has a winning strategy in the strong existential $k$-pebble game 
 on $\prod_{1\leq i\leq n}(\G_i,\bar a_i)$ and $(\G,\bar b)$. 
   
 By replacing the notion of strong homomorphism $\Rightarrow$ with its approximation $\Rightarrow_k$, for a fixed $k > 1$, 
 we can then define the following relaxed test:
 
\begin{itemize} 
\item \underline{$k$-pebble QBE test for CRPQs:} Takes as input a 
graph database $\G$ and $n$-ary relations $S^+$ and $S^-$ over $\G$. It 
accepts iff $\prod_{\bar a \in S^+} (\G,\bar a) \not\Rightarrow_k (\G,\bar b)$ for each tuple $\bar b \in S^-$. 
\end{itemize} 

The $k$-pebble definability test for CRPQs is defined analogously. As in the case of CQs, 
these tests characterize the existence of CRPQs-explanations/definitions of treewidth 
at most $k$. Formally, the treewidth of a CRPQ $\gamma = \exists \bar y \bigwedge_{1 \leq i \leq m} L_i(x_i,y_i)$ 
is the treewidth of the undirected 
graph that contains as nodes the existentially quantified variables of $\gamma$, 
i.e., those in $\bar y$, and whose set of edges is $\{\{x_i,y_i\} \mid 1 \leq i \leq m, \, x_i \neq y_i\}$. 
We denote by $\CTW(k)$ the class of CRPQs of treewidth at most $k$ (for $k\geq 1$). Then: 

\begin{theorem}
\label{theo:relaxed-tests-crpq}
Fix $k\geq 1$. Consider a database $\G$ and $n$-ary relations $S^+$ and $S^-$ over $\G$. 
\begin{enumerate}
\item There is a $\CTW(k)$-explanation for $S^+$ and $S^-$ over $\G$ if and only if the $(k+1)$-pebble QBE test 
for CRPQs accepts $\G$, $S^+$ and $S^-$.
\item There is a $\CTW(k)$-definition for $S^+$ over $\G$ if and only if 
the $(k+1)$-pebble definability test 
for CRPQs accepts $\G$ and $S^+$.
\end{enumerate}
\end{theorem}

Using similar ideas as for the existential $k$-pebble game, it is possible to prove that the problem of checking whether $\prod_{1\leq i\leq n}(\G_i,\bar a_i)\Rightarrow_k(\G,\bar b)$, given $(\G_1,\bar a_1),\dots,(\G_n,\bar a_n)$ and $(\G,\bar b)$, can be solved in exponential 
time for each fixed $k > 1$. 
We then obtain that the $k$-pebble QBE/definability tests for CRPQs take exponential time, 
and from Theorem \ref{theo:relaxed-tests-crpq} that {\sc $\CTW(k)$-query-by-example} and 
{\sc $\CTW(k)$-definability} are in \exptime\ (same than for $\TW(k)$ as stated in Corollary \ref{coro:completeness}). We also obtain an exponential upper bound on the cost of evaluating a $\CTW(k)$-explanation (in case it exists):

\begin{proposition}
\label{prop:eval-crpq}
Fix $k \geq 1$. The following statements hold:  
\begin{enumerate} 
\item 
{\sc $\CTW(k)$-query-by-example} and {\sc $\CTW(k)$-definability} are in \exptime.
\item 
Moreover, in case that there is a $\CTW(k)$-explanation of $S^+$ and $S^-$ over $\G$, 
the evaluation $\gamma(\G)$ of one such explanation $\gamma$ over $\G$ 
can be computed in exponential time. 
\end{enumerate} 
\end{proposition}

\section{Future work}
\label{sec:conc} 

We have left some problems open. The most notable one is determining the precise 
complexity of QBE/definability for CRPQs (resp., CRPQs of bounded treewidth). We have only obtained upper bounds for 
these problems that show that they are no more difficult than for CQs, but proving matching lower bounds 
seems challenging.  

An interesting line for future research is studying what to do when no explanation/definition exists for a set of examples. 
In such cases one might want to compute 
a query that minimizes the ``error'', e.g., the number of misclassified examples. 
We plan to study whether the techniques presented in this paper can be extended to deal with such problems. 

\medskip
\noindent
{\footnotesize 
{\bf Acknowledgements:} We are grateful to Leonid Libkin for their helpful 
comments in earlier versions of the paper and to Timos Antonopoulos for enlightening discussions about the notion of definability. 
We are also indebted to the reviewers of this article who helped us improving 
the presentation. In particular, one of the revieweres identified a subtle but important issue with the 
QBE tests presented in the paper that would have gone unnoticed otherwise.  
Barcel\'o and Romero are funded by the Millennium Nucleus Center for Semantic Web Research under
Grant NC120004. Romero is also funded by a Conicyt PhD scholarship.}

\bibliography{icdt2016}

\end{document}